\documentclass[11pt,reqno]{amsart}
\usepackage{txfonts}
\usepackage{mathrsfs}
\usepackage{amsfonts}
\allowdisplaybreaks[4]
\usepackage{graphicx} 
\usepackage{epsfig}
\usepackage{amssymb}
\usepackage{amsmath,xypic}
\usepackage{cite,color}
\usepackage{enumerate}
\usepackage{bm}
\usepackage{physics}
\usepackage{amsmath}
\usepackage{tikz}
\usepackage{mathdots}
\usepackage{yhmath}
\usepackage{cancel}
\usepackage{color}
\usepackage{siunitx}
\usepackage{array}
\usepackage{multirow}
\usepackage{amssymb}
\usepackage{gensymb}
\usepackage{tabularx}
\usepackage{extarrows}
\usepackage{booktabs}
\usetikzlibrary{fadings}
\usetikzlibrary{patterns}
\usetikzlibrary{shadows.blur}
\usetikzlibrary{shapes}

\newtheorem{theorem}{Theorem}

\newtheorem{proposition}[theorem]{Proposition}
\newtheorem{corollary}[theorem]{Corollary}

\newtheorem{lemma}[theorem]{Lemma}
\newtheorem{remark}[theorem]{Remark}

\newcommand{\be}{\begin{equation}}
\newcommand{\ee}{\end{equation}}
\newcommand{\bea}{\begin{eqnarray}}
\newcommand{\eea}{\end{eqnarray}}
\newcommand{\ba}{\begin{array}}
	\newcommand{\ea}{\end{array}}
\newcommand{\bean}{\begin{eqnarray*}}
	\newcommand{\eean}{\end{eqnarray*}}

\newcommand{\bt}{\bar{t}}

\newcommand{\tone}{t^{(1)}}
\newcommand{\ttwo}{t^{(2)}}

\numberwithin{equation}{section}
\renewcommand{\bt}{\bm{t}}

\begin{document}
\title{Toda Darboux transformations and vacuum expectation values }
\author{Chengwei Wang$^1$, Mengyao Chen$^1$, Jipeng Cheng$^{1,2*}$}
\dedicatory {$^1$ School of Mathematics, China University of
Mining and Technology, Xuzhou, Jiangsu 221116, P.\ R.\ China\\
$^2$ Jiangsu Center for Applied Mathematics (CUMT), Xuzhou, Jiangsu 221116, P.\ R.\ China}
\thanks{*Corresponding author. Email: chengjp@cumt.edu.cn, chengjipeng1983@163.com.}
\begin{abstract}
Determinant formulas for vacuum expectation values 
$\langle s+k+n-m,-s|e^{H(\bt)}\beta_m^{*}\cdots\beta_1^{*}\beta_n\cdots\beta_1g|k\rangle $ are given by using Toda Darboux transformations. Firstly notice that 2--Toda hierarchy can be viewed as the 2--component bosonizations of fermionic KP hierarchy, then two elementary Toda Darboux transformation operators $T_{+}(q)=\Lambda(q)\cdot\Delta\cdot q^{-1}$ and
$T_{-}(r)=\Lambda^{-1}(r)^{-1}\cdot\Delta^{-1}\cdot r$ are constructed from the changes of Toda (adjoint) wave functions by using 2--component boson--fermion correspondence. Based on this, the above vacuum expectation values now can be realized as the successive applications of Toda Darboux transformations. So the corresponding determinant formulas can be derived from the determinant representations of Toda Darboux transformations. Finally by similar methods, we also give the determinant formulas for $\langle n-m|e^{\mathcal{H}(\bm{x})}\beta_m^{*}\cdots\beta_1^{*}\beta_n\cdots\beta_1g|k\rangle $ related with KP tau functions. \\
\textbf{Keywords}:  \ Toda hierarchy,\ \ Darboux transformation,\ \ vacuum expectation values,\ \ tau function.\\
\textbf{MSC 2020}: 35Q53, 37K10, 37K40\\
\textbf{PACS}: 02.30.Ik
\end{abstract}

\maketitle

\section{Introduction}
Free fermions \cite{Alexandrov2013,Jimbo1983,Ten1991,Kac2003,Miwa2000,Kac1998,Kac2013,Kac1990} have been widely used in theoretical and mathematical physics, representations of infinite dimensional Lie algebras and other areas.  
Here we would like to discuss the charged free fermions \cite{Alexandrov2013,Miwa2000,Jimbo1983,Kac2013} $\psi_i$ and $\psi_i^{*}$ satisfying
\begin{equation}
	\psi_i\psi_j+\psi_j\psi_i=\psi_i^{*}\psi_j^{*}+\psi_j^{*}\psi_i^{*}=0,\ \psi_i\psi_j^{*}+\psi_j^{*}\psi_i=\delta_{ij}.\label{psirelation}
\end{equation}
And the fermionic Fock space $\mathcal{F}$ and its dual space $\mathcal{F}$ are defined by $\mathcal{F}=\mathcal{A}|0\rangle$ and $\mathcal{F}^{*}=\langle 0|\mathcal{A}$
with Clifford algebra $\mathcal{A}$ generated by $\psi_i$, $\psi_i^{*}$ and $1$. Here vacuums $|0\rangle$ and $\langle0|$  are defined by
\begin{align}
	&\psi_i|0\rangle=0\ (i<0),\ \psi_i^{*}|0\rangle=0 \ (i\geq0),\notag\\
	&\langle 0|\psi_i=0\ (i\geq0),\ \langle0| \psi_i^{*}=0\ (i<0).\label{vacuumcondition}
\end{align}
If define  $ \text{charge of}\ \psi_i=1$ and $\text{charge of}\ \psi_i^{*}=-1,$
 then $\mathcal{A}=\oplus_{k\in\mathbb{Z}}\mathcal{A}_k$ with $\mathcal{A}_k$ having charge  $k$, $\mathcal{F}=\oplus_{k\in\mathbb{Z}}\mathcal{F}_k$ and $\mathcal{F}^{*}=\oplus_{k\in\mathbb{Z}}\mathcal{F}_k^{*}$
 with $\mathcal{F}_k=\mathcal{A}_k|0\rangle=\mathcal{A}_0|k\rangle$ and $\mathcal{F}^*_k=\langle 0|\mathcal{A}_k=\langle k| \mathcal{A}_0 $,
 where $|k\rangle=\Psi_k|0\rangle$ and $\langle k|=\langle0|\Psi_k^{*}$ with
 $\Psi_{k}=\psi_{k}^{*}\dots\psi_{-1}^{*}$, $\Psi_{k}^*=\psi_{-1}\dots\psi_{k}$ for $k<0$, $\Psi_{k}=\psi_{k-1}\dots\psi_{0}$, $\Psi_{k}^*=\psi_{0}^{*}\dots\psi_{k-1}^{*}$ for $k>0$ and $\Psi_{0}=\Psi_{0}^*=1$.

For $\langle 0|a\in\mathcal{F}^*$ and $b|0\rangle\in\mathcal{F}$, the paring is denoted by $\langle 0|ab|0\rangle$, which can be determined by \eqref{psirelation} \eqref{vacuumcondition} and $\langle 0|1|0\rangle=1$. Here $\langle 0|a|0\rangle$ is called the vacuum expectation value (VEV) of $a\in\mathcal{A}$, which plays an important role in theoretical and mathematical physics \cite{Harnad2021,Harnad2006,Harnadarxiv,Alexandrov2020,
Harnad20212,Wangzy2023}. For general VEV, one can compute by Wick theorem \cite{Jimbo1983,Miwa2000} given below. 
For $v_1,\cdots v_r\in \bigoplus_{i\in\mathbb{Z}}\mathbb{C}\psi_i\oplus
\bigoplus_{i\in\mathbb{Z}}\mathbb{C}\psi_i^*$,
\begin{equation*}
\langle 0| v_1\cdots v_r|0\rangle=
\left\{
  \begin{array}{ll}
  0,&\text{$r$ is odd},\\
  \sum_{\eta\in S_n}{\rm sign}(\eta)\langle 0| v_{\eta(1)}v_{\eta(2)}|0\rangle\cdots\langle0| v_{\eta(r-1)}v_{\eta(r)}|0\rangle,&\text{$r$ is even},
\end{array}
\right.
\end{equation*}
where $\eta$ is the permutation such that $\eta(1)<\eta(2),\cdots,\eta(r-1)<\eta(r)$ and $ \eta(1)<\eta(3)<\cdots<\eta(r-1)$. Though by Wick theorem one can theoretically compute all VEVs, we do not believe Wick theorem is very applicable in practice. Actually vacuum expectation values are usually connected with some determinants, which can be better understood, compared with Wick theorem. Note that the general case $\langle 0| v_1\cdots v_r|0\rangle$ can be written into the sum of the following forms of VEVs
\begin{align}\label{kpvev}
\langle n-m|\beta_m^{*}\cdots\beta_1^{*}\beta_n\cdots\beta_1|0\rangle,\quad 
\beta_j\in \oplus_{i\in\mathbb{Z}}\mathbb{C}\psi_i,\quad\beta_j^*\in \oplus_{i\in\mathbb{Z}}\mathbb{C}\psi_i^*,\quad m,n\in\mathbb{Z}_{\geq 0},
\end{align}
since $\langle 0|a|0\rangle=0$ for $a\in\mathcal{A}_k$ with $k\neq 0$. To our best knowledge, there are determinant formulas for cases $n=m$, $n=0$ and $m=0$, which can be found in \cite{Alexandrov2013}, while for general cases, there are still no similar results. 

In order to make our discussion more typical, we would like to discuss the following VEV
\begin{align}\label{Todavev}
\langle s+k+n-m,-s|e^{H(\bt)}\beta_m^{*}\cdots\beta_1^{*}\beta_n\cdots\beta_1g|k\rangle,\quad
s,k\in\mathbb{Z},\quad n,m\in\mathbb{Z}_{\geq 0},
\end{align}
where  $\bm{t}=\left(t^{(1)},t^{(2)}\right)$, $t^{(a)}=(t_1^{(a)},t_2^{(a)},\cdots)$, $\langle l_1, l_2|=\langle 0 |\Psi_{l_1}^{(1)*}\Psi_{l_2}^{(2)*}$ with $\Psi_{l}^{(a)*}$ defined by replacing $\psi_i$ and $\psi_i^*$ in $\Psi_{l}$ with
$\psi_i^{(a)}$ and $\psi_i^{(a)*}$, $g\in GL_\infty=\{e^{X_1}\cdots e^{X_k}|X_i=\sum c_{kl}^{(i)}\psi_k\psi_l^{*}\}$  and $$H(\bm{t})=\sum_{a=1,2}\sum_{j\geq0}\sum_{k\in\mathbb{Z}}t_j^{(a)}:\psi_{k}^{(a)}\psi_{k+j}^{(a)*}:$$  with $:\psi_{i}^{(a)}\psi_j^{(b)*}:=\psi_{i}^{(a)}\psi_j^{(b)*}$ for $j\geq 0$ and $:\psi_i^{(a)}\psi_j^{(b)*}:=-\psi_j^{(b)*}\psi_i^{(a)}$ for $j<0$. Here $\psi_i^{(a)}$ and $\psi_j^{(b)*}$ comes from 
the rearrangement of labels for $\psi_i$ and $\psi_j^{*}$ such that 
\begin{align*}
\psi_i^{(a)}\psi_j^{(b)}+\psi_j^{(b)}\psi_i^{(a)}
=\psi_i^{(a)*}\psi_j^{(b)*}+\psi_j^{(b)*}\psi_i^{(a)*}=0,\quad \psi_i^{(a)}\psi_j^{(b)*}+\psi_j^{(b)*}\psi_i^{(a)}=\delta_{ab}\delta_{ij}.
\end{align*} 
For example, we can take $\psi_{j}^{(1)}=\psi_{2j}$, $\psi_{j}^{(2)}=\psi_{2j+1}$, $\psi_{j}^{(1)*}=\psi_{2j}^{*}$, 
$\psi_{j}^{(2)*}=\psi_{2j+1}^{*}$. We would remark that VEV \eqref{Todavev} corresponds to tau functions of 2--Toda hierarchy\cite{Takasaki2018,Jimbo1983}, which is the first non--trivial multi--component generalization of KP hierarchy, while general multi--component KP \cite{Kac2003,Kac1998} is corresponding to VEV
\begin{align}\label{p-KPvev}
\langle l_1,l_2,\cdots,l_p|a|0\rangle,\quad l_i\in\mathbb{Z},\quad a\in\mathcal{A}_{k},\quad k=\sum_{i=1}^pl_i,
\end{align}
where $\langle l_1,l_2,\cdots,l_p|$ is defined similarly to $\langle l_1,l_2|$. In what follows, Toda hierarchy always means 2--Toda hierarchy. 
Therefore if we can firstly understand VEV \eqref{Todavev} well, then it will be helpful to investigate VEV \eqref{p-KPvev} which also contains VEV \eqref{kpvev} as one special case.

To derive determinant formulas for  VEV \eqref{Todavev}, our strategy is to use Toda Darboux transformation to interpret VEV \eqref{Todavev}. Explicitly, 2--Toda hierarchy can be viewed as 2--component bosonization of the fermionic KP \cite{Alexandrov2013,Jimbo1983,Kac2003,Kac1998} hierarchy 
\begin{equation}\label{ferKP}
	S(\tau\otimes\tau)=0,\quad S=\sum_{j\in\mathbb{Z}}\psi_j\otimes\psi_j^{*},\quad \tau\in\mathcal{F}.
\end{equation}
One can refer to Section \ref{sec:2} for more details. In \cite{Liu2024}, it is proved that there exits a unique $k\in\mathbb{Z}$ such that $\tau\in\mathcal{F}_k$ for any $\tau\in\mathcal{F}$ satisfying the fermionic KP hierarchy. Further we can find $\tau=g|k\rangle$ with $g\in GL_{\infty}$ (see Chapter 9 in \cite{Miwa2000}).
For any	$\beta\in\oplus_{j\in\mathbb{Z}}\mathbb{C}\psi_{j}$ and $\beta^{*}\in\oplus_{j\in\mathbb{Z}}\mathbb{C}\psi_{j}^{*}$, if we denote \cite{Yang2022,Willox1998}
\begin{align*}
  \tau^{[1]}=\beta\tau\in\mathcal{F}_{k+1}\quad\text{or}
  \quad\beta^{*}\tau\in\mathcal{F}_{k-1},
\end{align*}
then $\tau^{[1]}$ is also the solution of fermionic KP hierarchy, that is $S(\tau^{[1]}\otimes\tau^{[1]})=0$. We will call the transformation $\tau\rightarrow \tau^{[1]}$ to be elementary Darboux transformation of fermionic KP hierarchy. Then
\begin{align}\label{multi-toda-tau-trans}
  \beta_m^{*}\cdots\beta_1^{*}\beta_n\cdots\beta_1\tau\in \mathcal{F}_{k+n-m}
\end{align}
is also the solution of fermionic KP hierarchy, which can be viewed as the multi--step Darboux transformation for fermionic KP. 
Then by 2--component boson--fermion correspondence, we can obtain VEV \eqref{Todavev}, which is corresponding to the transformations for Toda tau functions.  
The next goal is to derive the transformations for Toda Lax operators or Toda dressing operators corresponding to transformations \eqref{multi-toda-tau-trans} in tau functions, so that we can use the determinant formulas for Toda Darboux transformations to derive the ones for VEV \eqref{Todavev}. For this, we find in this paper elementary Toda Darboux transformation operators $$T_{+}(q)=\Lambda(q)\cdot\Delta\cdot q^{-1},\quad T_{-}(r)=\Lambda^{-1}(r)^{-1}\cdot\Delta^{-1}\cdot r$$
on Toda Lax or dressing operators, corresponding to transformations $\tau\rightarrow \beta\tau$ or $\tau\rightarrow \beta^*\tau$ respectively, where the key is the changes of Toda (adjoint) wave functions. Then based upon this, determinant formulas for VEV \eqref{Todavev} can be derived.  Similarly, we also use KP Darboux transformation to derive the determinant formulas for $\langle n-m|e^{\mathcal{H}(\bm{x})}\beta_m^{*}\cdots\beta_1^{*}\beta_n\cdots\beta_1g|k\rangle $ related with KP tau functions.

This paper is organized in the way below. In Section \ref{sec:2}, we review the construction of 2--Toda hierarchy from fermionic KP hierarchy by 2--component boson--fermion correspondence. Further in Section \ref{sec:3}, Darboux transformations for the whole 2--Toda hierarchy are established from the changes of Toda wave functions under the transformations of Toda tau functions. After that in Section \ref{sec:4}, determinant formulas of vacuum expectation values are given by using Toda Darboux transformations. At last, some conclusions and discussions are given in Section \ref{sec:5}.    

\section{Toda hierarchy and its fermionic construction}\label{sec:2}
In this section, we will review the 2--component boson--fermion correspondence\cite{Jimbo1983,Ten1991,Kac2003} and construct Toda hierarchy from fermionic KP hierarchy. For more details on Toda hierarchy, one can refer to \cite{Ueno1984,Takasaki2018}.

Firstly 2--component boson--fermion correspondence $\sigma_{Q,\bm{t}}: \mathcal{F}\cong\mathcal{B}=\bigoplus_{l_1,l_2\in\mathbb{Z}}\mathcal{C}Q_1^{l_1}Q_2^{l_2}\otimes\mathbb{C}\left[ \bm{t}\right]$
is defined by 
	\begin{equation*}\label{bos-fer}
		\sigma_{Q, \bm{t}}(a|0 \rangle )=\sum_{l_1, l_2\in\mathbb{Z}}(-1)^{l_1l_2}Q_{1}^{l_1}Q_{2}^{l_2}\langle l_1, l_2|e^{H(\bm{t})}a|0\rangle,\
		a\in\mathcal{A},
	\end{equation*}
where $Q_i$  commutes with $t_j^{(a)}$ satisfying $Q_1Q_2=-Q_2Q_1$.
If introduce fermionic fields 
\begin{equation}\label{defpsi(1)and*}
\psi^{(a)}(z)=\sum_{j\in\mathbb{Z}}\psi_j^{(a)}z^j,\ 
\psi^{(a)*}(z)=\sum_{j\in\mathbb{Z}}\psi_j^{(a)*}z^{-j},\ a=1,2,
\end{equation}
then
\begin{align*}
	&\sigma_{Q,\bm{t}}\psi^{(a)}(z)\sigma^{-1}_{Q,\bm{t}}=Q_az^{Q_a\partial_{Q_a}}e^{\xi(t^{(a)},z)}e^{-\xi(\widetilde{\partial^{(a)}},z^{-1})},\\
	&\sigma_{Q,\bm{t}}\psi^{(a)*}(z)\sigma^{-1}_{Q,\bm{t}}=Q_a^{-1}z^{1-Q_a\partial_{Q_a}}e^{-\xi(t^{(a)},z)}e^{\xi(\widetilde{\partial^{(a)}},z^{-1})},
\end{align*}
where $\xi(t^{(a)},z)=\sum_{j\geq1}t^{(a)}_jz^j$ and $\widetilde{\partial^{(a)}}=(\partial_{t_1^{(a)}},\partial_{t_2^{(a)}}/2,\cdots)$.
Moreover, we have the following relations \cite{Jimbo1983},
	\begin{equation*}
		\begin{aligned}
			&\langle l_1, l_2|\psi^{(1)}(z)e^{H(\bm{t})}=(-1)^{l_2}z^{l_1-1}\langle l_1-1, l_2|e^{H(\bm{t}-[z^{-1}]_1)}, \\
			&\langle l_1, l_2|\psi^{(2)}(z)e^{H(\bm{t})}=z^{l_2-1}\langle l_1, l_2-1|e^{H(\bm{t}-[z^{-1}]_2)}, \\	
			&\langle l_1, l_2|\psi^{(1)*}(z)e^{H(\bm{t})}=(-1)^{l_2}z^{-l_1}\langle l_1+1, l_2|e^{H(\bm{t}+[z^{-1}]_1)}, \\
			&\langle l_1, l_2|\psi^{(2)*}(z)e^{H(\bm{t})}=z^{-l_2}\langle l_1, l_2+1|e^{H(\bm{t}+[z^{-1}]_2)},	
		\end{aligned}
	\end{equation*}
where $[z^{-1}]=(z^{-1},z^{-2}/2,\cdots,z^{-n}/n,\cdots)$, $\bm{t}\pm[z^{-1}]_1=(t^{(1)}\pm[z^{-1}],t^{(2)})$ and $\bm{t}\pm[z^{-1}]_2=(t^{(1)},t^{(2)}\pm[z^{-1}])$.

After the above preparation, now we can rewrite fermionic KP hierarchy $S(\tau\otimes\tau)=0$ with $\tau=g|k\rangle$ into
\begin{equation}\label{respz}
{\rm Res}_zz^{-1}\Big(	\sum_{a=1}^{2}\psi^{(a)}(z)\otimes\psi^{(a)*}(z) \Big)\left(g|k\rangle\otimes g|k\rangle \right)=0,
\end{equation}
where ${\rm Res}_z\sum_ia_iz^{i}=a_{-1}$. Therefore, if applying $\langle s+k+1,-s|e^{H(\bm{t})}\otimes\langle s^{\prime}+k-1,-s^{\prime}|e^{H(\bm{t}^{\prime})} $ to \eqref{respz}, we have
\begin{equation}\label{bilinearequation}
	\begin{aligned}
		&\oint_{C_R}\frac{{\rm d}z}{2\pi i}\tau(s,\bm{t}-[z^{-1}]_1)\tau(s^{\prime},\bt^{\prime}+[z^{-1}]_1)z^{s-s^{\prime}}e^{\xi({\tone}-t^{\prime}, z)}\\
		=&\oint_{C_r}\frac{{\rm d}z}{2\pi i}\tau(s+1,\bt-[z]_2)\tau(s^{\prime}-1,\bt^{\prime}+[z]_2)z^{s-s^{\prime}}e^{\xi({{\ttwo}}-{{\ttwo}}^{\prime}, z^{-1})}, 
	\end{aligned}
\end{equation}
where $C_R$ means the anticlockwise circle $|z|=R$ for sufficient large $R$, while $C_r$ means the anticlockwise circle
$|z|=r$ for sufficient small $r$. And the tau function $\tau(s,\bm{t})$ is defined by
\begin{equation*}
	\tau(s,\bt)=(-1)^{\frac{s(s-1)}{2}}\langle s+k, -s|e^{H(\bt)}g|k\rangle.
\end{equation*}
If introduce wave function $\psi_i(s,\bm{t},z)$ and adjoint wave function $\psi_i^{*}(s,\bm{t},z)$ by
\begin{equation}\label{defwav}
	\begin{aligned} 
		&\psi_1(s,\bt, z)=\frac{ \tau(s,\bt-[z^{-1}]_1) }{ \tau(s,\bt)  }z^{s}e^{\xi({\tone}, z)}=(-1)^{s}z^{-k}\frac{\langle s+k+1, -s|e^{H(\bt)}\psi^{(1)}(z)g|k\rangle}{\langle s+k, -s |e^{H(\bt)}g|k\rangle},  \\
		&\psi_2(s,\bt, z)=  \frac{ \tau(s+1,\bt-[z]_2) }{ \tau(s,\bt)}z^{s}e^{\xi({{\ttwo}}, z^{-1})}=(-1)^{s}z^{-1}\frac{\langle s+k+1, -s|e^{H(\bt)}\psi^{(2)}(z^{-1})g|k\rangle}{\langle s+k, -s |e^{H(\bt)}g|k\rangle},\\
		&\psi_1^{*}(s,\bt, z)=\frac{ \tau(s+1,\bt+[z^{-1}]_1) }{ \tau(s+1,\bt)  }z^{-s}e^{-\xi({\tone}, z)}=(-1)^{s+1}z^{k}\frac{\langle s+k, -s-1|e^{H(\bt)}\psi^{(1)*}(z)g|k\rangle}{\langle s+k+1, -s-1|e^{H(\bt)}g|k\rangle}, \\
		&\psi_2^{*}(s,\bt, z)=\frac{ \tau(s,\bt+[z]_2) }{ \tau(s+1,\bt)  }z^{-s}e^{-\xi({{\ttwo}}, z^{-1})}=(-1)^{s}z\frac{\langle s+k, -s-1|e^{H(\bt)}\psi^{(2)*}(z^{-1})g|k\rangle}{\langle s+k+1, -s-1 |e^{H(\bt)}g|k\rangle},
	\end{aligned}
\end{equation}
then \eqref{bilinearequation} will become
\begin{equation}\label{wavbilinear}
	\oint_{C_R}\frac{{\rm d}z}{2\pi iz}\psi_1(s,\bt,z)\psi_1^{*}(s^{\prime},\bt^{\prime},z)=	\oint_{C_r}\frac{{\rm d}z}{2\pi iz}\psi_2(s,\bt,z)\psi_2^{*}(s^{\prime},\bt^{\prime},z).
\end{equation}

Further introduce wave operators $S_i$, $\widetilde{S}_i$, $W_i$ and $\widetilde{W}_i$ as follows
\begin{align*}
	&\qquad \quad S_1=1+\sum_{i=1}^{\infty}f_i(s,\bm{t})\Lambda^{-i},\
	S_2=g_0(s,\bm{t})+\sum_{i=1}^{\infty}g_i(s,\bm{t})\Lambda^{i},\\
	&\qquad\quad\widetilde{S}_1=1+\sum_{i=1}^{\infty}\widetilde{f}_i(s,\bt)\Lambda^{i},\ \ \
	\widetilde{S}_2=\widetilde{g}_0(s,\bt)+\sum_{i=1}^{\infty}\widetilde{g}_i(s,\bt)\Lambda^{-i},\\
		&W_1=S_1e^{\xi(t^{(1)},\Lambda)},\  W_2=S_2e^{\xi(\ttwo,\Lambda^{-1})},\
	\widetilde{W}_1=\widetilde{S}_1e^{-\xi(\tone,\Lambda^{-1})},\
	\widetilde{W}_2=\widetilde{S}_2e^{-\xi(\ttwo,\Lambda)},
\end{align*}
by requiring 
\begin{equation}\label{defpW}
	\psi_i(s,\bt,z)=W_i(z^{s}),\ \psi_i^{*}(s,\bt,z)=\widetilde{W}_i(z^{-s}),
\end{equation}
where $\Lambda$ is shift operator defined by  $\Lambda(f(s))=f(s+1)$. Then
one can find $\widetilde{S}_i=(S_i^{-1})^{*}$  with $(\sum_ia_i(s)\Lambda^i)^{*}=\sum_i\Lambda^{-i}a_i(s)$ and
\begin{align*}
	&\partial_{t_{p}^{(1)}}S_1=-\left(S_1\Lambda^{p}S_1^{-1} \right)_{<0}S_1,\ \partial_{t_{p}^{(2)}}S_1=\left(S_2\Lambda^{-p}S_2^{-1} \right)_{<0}S_1,\\
	&\partial_{t_{p}^{(1)}}S_2=\left(S_1\Lambda^{p}S_1^{-1} \right)_{\geq0}S_2,\ \ \ \ \partial_{t_{p}^{(2)}}S_2=-\left(S_2\Lambda^{-p}S_2^{-1} \right)_{\geq0}S_2.
\end{align*}
Base upon these relations, we have
\begin{equation}\label{parpi}
	\begin{aligned}
		&\partial_{t_p^{(1)}}\psi_i=\left((L_1^{p})_{\geq0}\right)(\psi_i),\quad \ \ \  \partial_{t_p^{(2)}}\psi_i=\left((L_2^{p})_{<0}\right)(\psi_i),\\
		&\partial_{t_p^{(1)}}\psi_i^{*}=-\left((L_1^{p})_{\geq0}\right)^{*}(\psi_i^{*}),\ \partial_{t_p^{(2)}}\psi_i^{*}=-\left((L_2^{p})_{<0}\right)^{*}(\psi_i^{*}).
	\end{aligned}
\end{equation}

Next if denote Lax operators by
\begin{equation*}
	L_1=S_1\Lambda S_1^{-1}=\Lambda+\sum_{i=1}^{\infty}u_i(s,\bm{t})\Lambda^{1-i},\ L_2=S_2\Lambda^{-1}S_2^{-1}=\sum_{i=0}^{\infty}v_i(s,\bt)\Lambda^{i-1},
\end{equation*}
then we have the following Lax equation
\begin{equation}\label{laxequ}
	\begin{aligned}
		&\partial_{t_p^{(1)}}L_1=\left[(L_1^{p})_{\geq0},L_1\right],\ \partial_{t_p^{(2)}}L_1=\left[(L_2^{p})_{<0},L_1\right],\\
		&\partial_{t_p^{(1)}}L_2=\left[(L_1^{p})_{\geq0},L_2\right],\ \partial_{t_p^{(2)}}L_2=\left[(L_2^{p})_{<0},L_2\right],\end{aligned}
\end{equation}
where  $\left(\sum_{j\in\mathbb{Z}}f_j\Lambda^{j}\right)_{\geq0}=\sum_{j\geq0}f_j\Lambda^j$ and $\left(\sum_{j\in\mathbb{Z}}f_j\Lambda^j\right)_{<0}=\sum_{j<0}f_j\Lambda^j$. Conversely, starting from Lax equation \eqref{laxequ}, we can prove $S_i$ and $W_i$ exist. Further the bilinear equation with wave functions \eqref{wavbilinear} can be obtained and finally we can show the tau functions exists. One can refer to \cite{Takasaki2018,Ueno1984} for more details. Therefore, all the above formulations for Toda hierarchy are equivalent.   

In what follows, we will try to find operator $T$ such that $L_i^{[1]}=TL_iT^{-1}$  still satisfies the Toda Lax equation \eqref{laxequ}, where $T$ is called the Toda Darboux transformation operator.  The key step is to find the changes in (adjoint) wave functions. If the transformed wave function $\psi_i^{[1]}$ or adjoint wave function $\psi_i^{*[1]}$
can be respectively expressed by difference operators $A$ and $B$ in the way below
\begin{equation}\label{keystep}
	\psi^{[1]}_i(s,\bt,z)=A\Big(z^{l_i}\psi_i(s,\bt,z)\Big),\
	\psi_i^{*[1]}(s,\bt,z)=B\Big(z^{m_i}\psi_i^{*}(s,\bt,z)\Big),
\end{equation}
where $l_i$ and $m_i$ are some integers such that $\psi_i^{[1]}$ and $\psi_i^{*[1]}$ have the same expansions in $z$ as $\psi_i$ and $\psi_i^{*}$ respectively.
Then by \eqref{defpW}, one can find $$W_i^{[1]}=AW_i \Lambda^{l_i},\
\widetilde{W}^{[1]}_i=B\widetilde{W}_i\Lambda^{-m_i},
$$ which implies  that $T=A$ or $T=(B^{-1})^{*}$.
We will show more details in the following section.

\section{Darboux transformations of Toda hierarchy} \label{sec:3}
In this section, we will firstly investigate elementary Toda  Darboux transformation based on (adjoint) wave functions. Then we prove transformed (adjoint) eigenfunction defined by fermionic transformation can be expressed by corresponding Toda Darboux transformation. 
Finally, we generalize to multi--step transformation. 

By using 2--component boson--fermion correspondence $\sigma_{Q,\bt}$, the transformed tau functions are given in the following way,
\begin{itemize}
	\item for $\tau^{[1]}=\beta\tau$,
	\begin{equation*}
    \tau^{[1]}(s,\bt)=(-1)^{\frac{s(s-1)}{2}}\langle s+k+1,-s|e^{H(\bt)}\beta g|k\rangle=(-1)^{s}q(s,\bt)\tau(s,\bt),
	\end{equation*}
\end{itemize}
\begin{itemize}
	\item for $\tau^{[1]}=\beta^{*}\tau$,
	\begin{equation*}
		\tau^{[1]}(s,\bt)=(-1)^{\frac{s(s-1)}{2}}\langle s+k-1,-s|e^{H(\bt)}\beta^{*} g|k\rangle=(-1)^{s}r(s-1,\bt)\tau(s,\bt),
	\end{equation*}
\end{itemize}
where $q$ and $r$ are defined by
\begin{equation}\label{defqr}
	q(s,\bt)=(-1)^{s}\frac{\langle s+k+1, -s|e^{H(\bt)}\beta g|k\rangle}{\langle s+k, -s|e^{H(\bt)}g|k\rangle},\ 
	r(s,\bt)=(-1)^{s+1}\frac{\langle s+k, -s-1|e^{H(\bt)}\beta^{*}g|k\rangle}{\langle s+k+1, -s-1|e^{H(\bt)}g|k\rangle}.
\end{equation}
Then by comparing with \eqref{defwav}, the corresponding transformed (adjoint) wave functions can be written as
\begin{itemize}
	\item for $\tau^{[1]}=\beta\tau$,
	\begin{equation}\label{traadjwav1}
			\begin{aligned}
			&\psi_1^{[1]}(s,\bt,z)=(-1)^{s}z^{-k-1}\frac{\langle s+k+2,-s|e^{H(\bt)}\psi^{(1)}(z)\beta g|k\rangle }{\langle s+k+1,-s|e^{H(\bt)}\beta g|k\rangle }=\frac{q(s,\bt-[z^{-1}]_1) }{q(s,\bt) }\psi_1(s,\bt,z),\\
			&\psi_2^{[1]}(s,\bt,z)=(-1)^{s}z^{-1}\frac{\langle s+k+2,-s|e^{H(\bt)}\psi^{(2)}(z^{-1})\beta g|k\rangle
			 }{\langle s+k+1,-s|e^{H(\bt)}\beta g|k\rangle }=-\frac{q(s+1,\bt-[z]_2) }{q(s,\bt) }\psi_2(s,\bt,z),
			\\
			&\psi_1^{*[1]}(s,\bt,z)=(-1)^{s+1}z^{k+1}\frac{\langle s+k+1, -s-1|e^{H(\bt)}\psi^{(1)*}(z)\beta g|k\rangle}{\langle s+k+2, -s-1 |e^{H(\bt)}\beta g|k\rangle}=\frac{q(s+1,\bt+[z^{-1}]_1)}{q(s+1,\bt)}\psi_1^{*}(s,\bt, z), \\
			&\psi_2^{*[1]}(s,\bt,z)=(-1)^{s}z\frac{\langle s+k+1, -s-1|e^{H(\bt)}\psi^{(2)*}(z^{-1})\beta g|k\rangle}{\langle s+k+2, -s-1 |e^{H(\bt)}\beta g|k\rangle}=-\frac{q(s,\bt+[z]_2)}{q(s+1,\bt)}\psi_2^{*}(s,\bt, z),
		\end{aligned}
	\end{equation}
\end{itemize}
\begin{itemize}
	\item for $\tau^{[1]}=\beta^{*}\tau$,
	\begin{equation}\label{traadjwav2}
		\begin{aligned}
			&\psi_1^{[1]}(s,\bt, z)=(-1)^{s}z^{-k+1}\frac{\langle s+k, -s|e^{H(\bt)}\psi^{(1)}(z)\beta^{*} g|k\rangle}{\langle s+k-1, -s |e^{H(\bt)} \beta^{*} g|k\rangle}=\frac{r(s-1,\bt-[z^{-1}]_1)}{r(s-1,\bt)}\psi_1(s,\bt,z), \\
		&\psi_2^{[1]}(s,\bt, z)=(-1)^{s}z^{-1}\frac{\langle s+k, -s|e^{H(\bt)}\psi^{(2)}(z^{-1})\beta^{*} g|k\rangle}{\langle s+k-1, -s |e^{H(\bt)}\beta^{*} g|k\rangle}=-\frac{r(s,\bt-[z]_2)}{r(s-1,\bt)}\psi_2(s,\bt,z),\\
		&\psi_1^{*[1]}(s,\bt,z)=(-1)^{s+1}z^{k-1}\frac{\langle s+k+1,-s-1|e^{H(\bt)}\psi^{(1)*}(z)\beta^{*}g|k\rangle}{\langle s+k,-s-1|e^{H(\bt)}\beta^{*}g|k\rangle}=\frac{r(s,\bt+[z^{-1}]_1) }{r(s,\bt) }\psi_1^{*}(s,\bt,z),\\
		&\psi_2^{*[1]}(s,\bt,z)=(-1)^{s}z\frac{\langle s+k+1,-s-1|e^{H(\bt)}\psi^{(2)*}(z^{-1})\beta^{*}g|k\rangle }{\langle s+k,-s-1|e^{H(\bt)}\beta^{*}g|k\rangle }=-\frac{r(s-1,\bt+[z]_2) }{r(s,\bt) }\psi_2^{*}(s,\bt,z).
		\end{aligned}
	\end{equation}
\end{itemize}
In order to write $\psi_i^{[1]}$ and $\psi_i^{*[1]}$ in the form of \eqref{keystep},  we need the following proposition.
\begin{proposition}\label{proresqr}
	Given $q$, $r$ defined in \eqref{defqr} and $\psi_i$, $\psi_i^{*}$ defined in \eqref{defwav},
	 \begin{align}
		q(s^{\prime},\bt^{\prime})=&\oint_{C_R}\frac{{\rm d}z}{2\pi i}
		z^{-2}q(s^{\prime\prime},\bt^{\prime\prime}+[z^{-1}]_1)\psi_1^{*}(s^{\prime\prime}-1,\bt^{\prime\prime},z)\psi_1(s^{\prime},\bt^{\prime},z)\nonumber\\
		&+\oint_{C_r}\frac{{\rm d}z}{2\pi i}z^{-1}q(s^{\prime\prime}-1,\bt^{\prime\prime}+[z]_2)\psi_2^{*}(s^{\prime\prime}-1,\bt^{\prime\prime},z)\psi_2(s^{\prime},\bt^{\prime},z),\label{phires}\\	
		r(s^{\prime},\bt^{\prime})=&\oint_{C_R}\frac{{\rm d}z}{2\pi i}z^{-2}r(s^{\prime\prime}-1,\bt^{\prime\prime}-[z^{-1}]_1)\psi_1(s^{\prime\prime},\bt^{\prime\prime},z)\psi_1^{*}(s^{\prime},\bt^{\prime},z)\nonumber\\
		&+\oint_{C_r}\frac{{\rm d}z}{2\pi i}z^{-1}r(s^{\prime\prime},\bt^{\prime\prime}-[z]_2)\psi_2(s^{\prime\prime},\bt^{\prime\prime},z)\psi_2^{*}(s^{\prime},\bt^{\prime},z).\label{psires}
		\end{align}
Further $q$ is Toda eigenfunction, while $r$ is adjoint Toda  eigenfunction, that is 
\begin{equation}\label{parqr}
	\begin{aligned}
		&\partial_{{t_p^{(1)}}}q=\left((L_1^{p})_{\geq0}\right)(q),\quad \ \partial_{{t_p^{(2)}}}q=\left((L_2^{p})_{<0}\right)(q),\\
		&\partial_{{t_p^{(1)}}}r=-\left((L_1^{p})_{\geq0}\right)^{*}(r),\ \partial_{{t_p^{(2)}}}r=-\left((L_2^{p})_{<0}\right)^{*}(r).
	\end{aligned}
\end{equation}
\end{proposition}
\begin{proof}
For $\tau=g|k\rangle\in\mathcal{F}_k$ satisfying fermionic KP hierarchy $S(\tau\otimes\tau)=0$, 
	\begin{equation*}
		S(\tau\otimes\beta\tau)=\beta\tau\otimes\tau,\
S(\beta^{*}\tau\otimes\tau)=\tau\otimes\beta^{*}\tau,\quad  \beta\in\oplus_{j\in\mathbb{Z}}\mathbb{C}\psi_{j},\ \beta^{*}\in\oplus_{j\in\mathbb{Z}}\mathbb{C}\psi_{j}^{*},
	\end{equation*}
where we have used the relations $S(1\otimes\beta)=\beta\otimes1-(1\otimes\beta)S$ and
$S(\beta^{*}\otimes1)=1\otimes\beta^{*}-(\beta^{*}\otimes1)S$
(see  \cite{Yang2022}). If acting  $\langle s^{\prime}+k+1,-s^{\prime}|e^{H(\bm{t}^{\prime})}\otimes\langle s^{\prime\prime}+k,-s^{\prime\prime}|e^{H(\bt^{\prime\prime})}$ on the first equation and  $\langle s^{\prime\prime}+k,-s^{\prime\prime}|e^{H(\bm{t}^{\prime\prime})}\otimes\langle s^{\prime}+k,-s^{\prime}-1|e^{H(\bt^{\prime})}$ on the second one, we can obtain \eqref{phires} and \eqref{psires} respectively. Finally \eqref{parqr} can be obtained by \eqref{parpi}.
\end{proof}
\begin{remark}
	Conversely, for Toda eigenfunction $q$ and adjoint eigenfunction $r$, we also have \eqref{phires} and \eqref{psires}, which is just the spectral representation of Toda eigenfunction and adjoint eigenfunction \cite{Liu2024}. 
\end{remark}
\begin{remark}\label{remToda}
	If denote $\tau_1(s,\bt)=(-1)^{s}q(s,\bt)\tau(s,\bt)$ and $\tau_{-1}(s,\bt)=(-1)^{s}r(s-1,\bt)\tau(s,\bt)$, then by \eqref{defwav}, \eqref{phires} and \eqref{psires}, we can find $\big(\tau(s,\bt),\tau_{1}(s,\bt)\big)$ and $\big(\tau_{-1}(s,\bt),\tau(s,\bt)\big)$ satisfy bilinear equation of modified Toda hierarchy (or 2--component 1st modified KP hierarchy), that is
	\begin{equation}\label{mTodahierarchy}
		\begin{aligned}
			&\oint_{C_R}\frac{{\rm d}z}{2\pi i}
			z^{s-s^{\prime}-1}\tau_{0,s}(\bt-[z^{-1}]_1)\tau_{1,s^{\prime}}(\bt^{\prime}+[z^{-1}]_1)e^{\xi(t^{(1)}-{t^{(1)}}^{\prime},z)}\\
			=&-\oint_{C_r}\frac{{\rm d}z}{2\pi i}
			z^{s-s^{\prime}}\tau_{0,s+1}(\bt-[z]_2)\tau_{1,s^{\prime}-1}(\bt^{\prime}+[z]_2)e^{\xi({t^{(2)}}-{t^{(2)}}^{\prime},z^{-1})}+\tau_{1,s}(\bt)\tau_{0,s^{\prime}}(\bt^{\prime}),
		\end{aligned}
	\end{equation}
	where $\big(\tau_{0,s}(\bt), \tau_{1,s}(\bt)\big)$ is corresponding modified Toda tau pair.
\end{remark}
\begin{lemma}[\cite{Liu2024}]
	\label{lemTodacor}
	Given $\big(\tau_{0,s}(\bt),\tau_{1,s}(\bt)\big)$ satisfying modified Toda hierarchy \eqref{mTodahierarchy},
	\begin{align*}
		&\tau_{1,s+1}(\bt+[z^{-1}]_1)\tau_{0,s}(\bt)-z\cdot \tau_{1,s}(\bt+[z^{-1}]_1)\tau_{0,s+1}(\bt)=-z\cdot \tau_{0,s+1}(\bt+[z^{-1}]_1)\tau_{1,s}(\bt),\\
		&\tau_{1,s}(\bt+[z]_2)\tau_{0,s}(\bt)-z\cdot \tau_{1,s-1}(\bt+[z]_2)\tau_{0,s+1}(\bt)=\tau_{0,s}(\bt+[z]_2)\tau_{1,s}(\bt).
	\end{align*}
\end{lemma}
Therefore according to Remark \ref{remToda} and Lemma \ref{lemTodacor}, we have the following Lemma.

\begin{lemma}Denote $\Delta=\Lambda-1$, $\iota_{\Lambda^{+}}\Delta^{-1}=-\sum_{i\geq0}\Lambda^{i}$ and $\iota_{\Lambda^{-}}\Delta^{-1}=\sum_{i\geq1}\Lambda^{-i}$,
	\begin{align}
		&\Delta\Big(q(s,\bt)^{-1}\cdot\psi_1(s,\bt,\lambda)\Big)=\lambda\frac{q(s,\bt-[\lambda^{-1}]_1) }{q(s+1,\bt)q(s,\bt) }\psi_1(s,\bt,\lambda),\label{detqip1}\\
		&\Delta\Big(q(s,\bt)^{-1}\cdot\psi_2(s,\bt,\lambda) \Big)=-\frac{q(s+1,\bt-[\lambda]_2) }{q(s,\bt)q(s+1,\bt) }\psi_2(s,\bt,\lambda),\label{detqip2}\\
		&\Delta\Big(r(s,\bt)^{-1}\cdot\psi_1^{*}(s,\bt,\lambda) \Big)=-\lambda\frac{r(s+1,\bt+[\lambda^{-1}]_1) }{r(s+1,\bt)r(s,\bt)}\psi_1^{*}(s+1,\bt,\lambda),\label{detrip1}\\
		&\Delta\Big(r(s,\bt)^{-1}\cdot\psi_2^{*}(s,\bt,\lambda)\Big)=\frac{r(s,\bt+[\lambda]_2) }{r(s+1,\bt)r(s,\bt) }\psi_2^{*}(s+1,\bt,\lambda),\label{detrip2}\\
		&\iota_{\Lambda^{+}}\Delta^{-1} \left(q(s,\bt)\cdot\psi_1^{*}(s,\bt,\lambda)\right)=-\lambda^{-1}q(s,\bt+[\lambda^{-1}]_1)\psi_1^{*}(s-1,\bt,\lambda),\label{deltqp1} \\
		&\iota_{\Lambda^{-}}\Delta^{-1}\left(q(s,\bt)\cdot\psi_2^{*}(s,\bt,\lambda)\right)=q(s-1,\bt+[\lambda]_2)\psi_2^{*}(s-1,\bt,\lambda),\label{deltqp2}\\
		&\iota_{\Lambda^{-}}\Delta^{-1}\Big(\psi_1(s,\bt,\lambda)\cdot r(s,\bt)\Big) 
		=\lambda^{-1}\psi_1(s,\bt,\lambda)r(s-1,\bt-[\lambda^{-1}]_1),\label{deltqp3}\\
		&\iota_{\Lambda^{+}}\Delta^{-1}\Big( \psi_2(s,\bt,\lambda)\cdot r(s,\bt)\Big)=-\psi_2(s,\bt,\lambda)r(s,\bt-[\lambda]_2),\label{deltqp4}
	\end{align}
	where $\psi_i$ and $\psi_i^{*}$ are defined by \eqref{defwav}, while $q$ and $r$ are defined by \eqref{defqr}.
\end{lemma}
By substituting \eqref{detqip1}--\eqref{deltqp4} into \eqref{traadjwav1}--\eqref{traadjwav2}, we get the following proposition.
\begin{proposition}\label{traphipsi}
	Given $\psi_i$, $\psi_i^{*}$ defined by \eqref{defwav}, and $q$, $r$, $\psi_i^{[1]}$ and $\psi_i^{*[1]}$  defined by \eqref{defqr}--\eqref{traadjwav2},
	 \begin{itemize}
	 	\item $\tau\rightarrow\tau^{[1]}=\beta\tau$,
	 	\begin{align*}
	 		&\psi_1^{[1]}(s,\bt,z)=z^{-1}\Big(q(s+1,\bt)\cdot\Delta\cdot q(s,t)^{-1}\Big)\Big(\psi_1(s,\bt,z)\Big),\\
	 		&\psi_2^{[1]}(s,\bt,z)=\Big(q(s+1,\bt)\cdot\Delta\cdot q(s,\bt)^{-1}\Big)\Big(\psi_2(s,\bt,z) \Big),\\
	 			&\psi_1^{*[1]}(s,\bt, z)=-z \Big(q(s+1,\bt)^{-1}\cdot\iota_{\Lambda^{+}}\Delta^{-1}\cdot q(s+1,\bt)\Big) \Big(\psi_1^{*}(s+1,\bt,z)\Big), \\
	 		&\psi_2^{*[1]}(s,\bt, z)=-\Big(q(s+1,\bt)^{-1}\cdot\iota_{\Lambda^{-}}\Delta^{-1}\cdot q(s+1,\bt)\Big)\Big(\psi_2^{*}(s+1,\bt,z)\Big),
	 	\end{align*}
	 \end{itemize}
	 \begin{itemize}
	 	\item $\tau\rightarrow\tau^{[1]}=\beta^{*}\tau$,
	 		\begin{align*}
	 		&\psi_1^{[1]}(s,\bt,z)=z\Big(r(s-1,\bt)^{-1}\cdot\iota_{\Lambda^{-}}\Delta^{-1}\cdot r(s,\bt)\Big)\Big(\psi_1(s,\bt,z)\Big), \\
	 		&\psi_2^{[1]}(s,\bt,z)=\Big(r(s-1,\bt)^{-1}\cdot\iota_{\Lambda^{+}}\Delta^{-1}\cdot r(s,\bt)\Big)\Big(\psi_2(s,\bt,z)\Big),\\
	 		&\psi_1^{*[1]}(s,\bt,z)=-z^{-1}\Big(r(s-1,\bt)\cdot\Delta\cdot r(s-1,\bt)^{-1}\Big)\Big(\psi_1^{*}(s-1,\bt,z) \Big),\\
	 		&\psi_2^{*[1]}(s,\bt,z)=-\Big(r(s-1,\bt)\cdot\Delta\cdot r(s-1,\bt)^{-1}\Big)\Big(\psi_2^{*}(s-1,\bt,z)\Big).
	 	\end{align*}
	 \end{itemize}
\end{proposition}
By now we have obtained the relations in \eqref{keystep}. Further using \eqref{defpW}, we can find Toda Darboux transformation operator $T$ stated in the following theorem.
\begin{theorem}
Given Toda eigenfunction $q$ and Toda adjoint  eigenfunction $r$ (see \eqref{parqr}),
\begin{equation*}
	T_{+}(q)=\Lambda(q)\cdot\Delta\cdot q^{-1},\
	T_{-}(r)=\Lambda^{-1}(r)^{-1}\cdot\Delta^{-1}\cdot r
\end{equation*} 
 are Toda Darboux transformation operators, which acts on Toda Lax operators $L_1$ and $L_2$ by the way below
	\begin{itemize}
		\item for $T_{+}(q)$,
		\begin{equation*}
			L_1^{[1]}=T_{+}(q)\cdot L_1\cdot\iota_{\Lambda^{-}}T_{+}(q)^{-1},\ L_2^{[1]}=T_{+}(q)\cdot L_2\cdot\iota_{\Lambda^{+}}T_{+}(q)^{-1},
		\end{equation*}
	\end{itemize}
	\begin{itemize}
		\item for $T_{-}(r)$,
		\begin{equation*}
			L_1^{[1]}=\iota_{\Lambda^{-}}T_{-}(r)\cdot L_1\cdot T_{-}(r)^{-1},\ L_2^{[1]}=\iota_{\Lambda^{+}}T_{-}(r)\cdot L_2\cdot T_{-}(r)^{-1},
		\end{equation*}
	\end{itemize}
	where $\iota_{\Lambda^{\pm}}T_{+}(q)^{-1}=q\cdot\iota_{\Lambda^{\pm}}\Delta^{-1}\cdot \Lambda(q)^{-1}$ and $\iota_{\Lambda^{\pm}}T_{-}(r)=\Lambda^{-1}(r)^{-1}\cdot\iota_{\Lambda^{\pm}}\Delta^{-1}\cdot r$.
\end{theorem}
\begin{corollary}\label{cortratsqrt}
	Given Toda Darboux transformation operators $T_{+}(q)$ and $T_{-}(r)$, their actions on corresponding wave operators $S_i$, wave function $\psi_i$, adjoint wave function $\psi_i^{*}$ and tau function $\tau(s,t)$ can be shown in the following table.
	\begin{center}
		\begin{tabular}{lll}
			\multicolumn{3}{c}{Table I. Darboux transformations: Toda $\rightarrow$ Toda}\\
			\hline \hline
			&$T_1=T_{+}(q)$ &\quad$T_2=T_{-}(r)$  \\
			\hline
			{\rm Wave operator} &$S_1^{[1]}=T_1\cdot S_1\cdot\Lambda^{-1}$& \quad$S_1^{[1]}=\iota_{\Lambda^{-}}T_2\cdot S_1\cdot\Lambda$\\
			&$S_2^{[1]}=T_1\cdot S_2$& \quad$S_2^{[1]}=\iota_{\Lambda^{+}}T_2\cdot S_2$\\
			{\rm Wave function}&$\psi_1^{[1]}=z^{-1}\cdot T_1(\psi_1)$& \quad$ \psi_1^{[1]}=z\cdot\iota_{\Lambda^{-}}T_2(\psi_1) $ \\
			&$\psi_2^{[1]}=T_1(\psi_2)$& \quad $\psi_2^{[1]}=\iota_{\Lambda^{+}}T_2(\psi_2)$ \\
			{\rm Adjoint wave function}&$\psi_1^{*[1]}=z\cdot\iota_{\Lambda^{+}}(T_1^{-1})^{*}(\psi_1^{*})$&\quad $\psi_1^{*[1]}=z^{-1}\cdot (T_2^{-1})^{*}(\psi_1^{*})$ \\
			& $\psi_2^{*[1]}=\iota_{\Lambda^{-}}(T_1^{-1})^{*}(\psi_2^{*})$ & \quad
			$\psi_2^{*[1]}=(T_2^{-1})^{*}(\psi_2^{*})$ \\
			{\rm Eigenfunction}  &  $\tilde{q}^{[1]}=T_1(\tilde{q})$    &\quad $\tilde{q}^{[1]}=T_2(\tilde{q})$ 
			\\
			{\rm Adjoint eigenfunction} &  $\tilde{r}^{[1]}=(T_1^{-1})^{*}(\tilde{r})$  &\quad $\tilde{r}^{[1]}=(T_2^{-1})^{*}(\tilde{r})$
			\\
			{\rm Tau function} &$\tau^{[1]}(s)=(-1)^{s}q(s)\tau(s)$& \quad$\tau^{[1]}(s)=(-1)^{s}r(s-1)\tau(s)$\\
			\hline
		\end{tabular}
	\end{center}
\end{corollary}
\begin{remark}\label{retraqr}
		For the changes of eigenfunction $\tilde{q}$, adjoint eigenfunction $\tilde{r}$, they can be obtained according to spectral representations \eqref{phires}--\eqref{psires} and changes of $\psi_i$ and $\psi_i^{*}$. 
\end{remark}
Note that in \cite{Song2023}, $T_{\pm}$ commutes with each other:
\begin{align*}
	T_{+}(q^{[1]})T_{-}(r)=T_-(r^{[1]})T_{+}(q),\quad
	T_{+}(q_2^{[1]})T_{+}(q_1)=T_{+}(q_1^{[1]})T_{+}(q_2),\quad
	T_{-}(r_2^{[1]})T_{-}(r_1)=T_{-}(r_1^{[1]})T_{-}(r_2).
\end{align*} 
Therefore when discussing the multi--step Toda Darboux transformations, we can only need to consider the following case:
\begin{equation*}
	\begin{aligned}
		&(L_1,L_2)
		\xrightarrow{T_{+}(q_1)} (L_1^{[1]},L_2^{[1]})
		\xrightarrow{T_{+}(q_{2}^{[1]})}(L_1^{[2]},L_2^{[2]})\longrightarrow\cdots\longrightarrow (L_1^{[n-1]},L_2^{[n-1]})\xrightarrow{T_{+}(q_{n}^{[n-1]})}(L_1^{[n]},L_2^{[n]})\\
		&\xrightarrow{T_{-}(r_{1}^{[n]})}(L_1^{[n+1]},L_2^{[n+1]})\xrightarrow{T_{-}(r_{2}^{[n+1]})}(L_1^{[n+2]},L_2^{[n+2]})\longrightarrow\cdots\longrightarrow 
		(L_1^{[n+m-1]},L_2^{[n+m-1]})\xrightarrow{T_{-}(r_{m}^{[n+m-1]})}(L_1^{[n+m]},L_2^{[n+m]}),
	\end{aligned}
\end{equation*}
where $q_i$ $(1\leq i\leq n)$ and $r_j$ $(1\leq j\leq m)$
are respectively $n$ independent Toda eigenfunctions and $m$ Toda adjoint eigenfunctions, and $A^{[n]}$ means the transformed object $A$ under $n$--step elementary Toda Darboux transformations.
If denote 
\begin{equation*}
	T^{[n,m]}=T_{-}(r_m^{[n+m-1]})\cdots T_{-}(r_1^{[n]})\cdots T_{+}(q_n^{[n-1]})\cdots T_{+}(q_1),
\end{equation*}
then $(n+m)$--step Darboux transformed Lax operators $L_1^{[n+m]}$ and $L_2^{[n+m]}$ can be expressed by
\begin{align*}
		L_1^{[n+m]}=\iota_{\Lambda^{-}}T^{[n,m]}\cdot L_1\cdot\iota_{\Lambda^{-}}(T^{[n,m]})^{-1},\
		L_2^{[n+m]}=\iota_{\Lambda^{+}}T^{[n,m]}\cdot L_2\cdot\iota_{\Lambda^{+}}(T^{[n,m]})^{-1}.
\end{align*}
Here $T^{[n,m]}$ is called $(n+m)$--step Toda Darboux transformation operator. By similar method in \cite{Liu2010}, one can obtain the determinant representation of $T^{[n,m]}$, which is given in Appendix. 

Further given Toda eigenfunction $q$, adjoint eigenfunction $r$ and tau function $\tau^{[0]}$, we have the corresponding transformed ones under $T^{[n,m]}$
given in the way below,
\begin{equation}\label{deftraqrtau}
	\begin{aligned}
		&q^{[n+m]}=T^{[n,m]}(q),\ r^{[n+m]}=\left((T^{[n,m]})^{-1}\right)^*(r),\\		&\tau^{[n+m]}=(-1)^{(n+m)s}\Lambda^{-1}(r_m^{[n+m-1]})\cdots\Lambda^{-1}(r_1^{[n]})q_n^{[n-1]}\cdots q_2^{[1]}q_1\tau^{[0]},
	\end{aligned}
\end{equation}
By determinant representation of $T^{[n,m]}$, one can further obtain the following results.
 \begin{itemize}
	\item $n\geq m$,
	\begin{align}
		&q^{[n+m]}=\frac{IW_{m,n+1}(r_m,\dots,r_1;q_1,\dots,q_n,q)}{IW_{m,n}(r_m,\dots,r_1;q_1,\dots,q_n)},\nonumber\\
		&r^{[n+m]}=(-1)^{n}\frac{\Lambda(IW_{m+1,n}(r,r_m,\dots,r_1;q_1,\dots,q_n))}{\Lambda(IW_{m,n}(r_m,\dots,r_1;q_1,\dots,q_n))},\nonumber\\
		&\tau^{[n+m]}=(-1)^{nm+(m+n)s}IW_{m,n}(r_m,\dots,r_1;q_1,\dots,q_n)\tau^{[0]},\label{taunmIwmn}
	\end{align}
\end{itemize}
\begin{itemize}
	\item $n<m$,
	\begin{align}
		&q^{[n+m]}=(-1)^{m}\frac{\Lambda^{-1}\left(IW^{*}_{n+1,m}(q,q_n,\dots,q_1;r_1,\dots,r_m)
			\right) }{\Lambda^{-1}\left(IW^{*}_{n,m}(q_n,\dots,q_1;r_1,\dots,r_m) \right) },\nonumber\\
		&r^{[n+m]}=\frac{IW_{n,m+1}^{*}(q_n,\dots,q_1;r_1,\dots,r_m,r)}{IW_{n,m}^{*}(q_n,\dots,q_1;r_1,\dots,r_m)},\nonumber\\
		&\tau^{[n+m]}=(-1)^{(n+m)s}\Lambda^{-1}\left(IW_{n,m}^{*}(q_n,\dots,q_1;r_1,\dots,r_m)\right)\tau^{[0]},\label{taunmIW*nm}
	\end{align}
\end{itemize}
where $IW_{i,j}$ and $IW^{*}_{i,j}$ are generalized discrete Wronskian determinants (see Appendix). 
\section{Determinant formulas of vacuum expectation values} \label{sec:4}
In this section, we will compute VEV $\langle s+k+n-m,-s|e^{H(\bt)}\beta_m^{*}\cdots\beta_1^{*}\beta_n\cdots\beta_1g|k\rangle $ by considering it as the transformed tau function $\tau^{[n+m]}$ obtained from $\tau^{[0]}=\langle s+k,-s|e^{H(\bt)}g|k\rangle$ under $T^{[n,m]}$, since we have obtained the determinant formulas for $\tau^{[n+m]}$ in the previous section.

Firstly in order to view above VEV as the transformed tau function under $T^{[n,m]}$, we need the following lemma.
\begin{lemma}\label{covacqr}
	For $\beta_i\in\oplus_{j\in\mathbb{Z}}\mathbb{C}\psi_j$, $\beta_i^{*}\in\oplus_{j\in\mathbb{Z}}\mathbb{C}\psi_j^{*}$ ($i=1,2$)
	and $g\in GL_{\infty}$, if $q_i$ and $r_i$ are defined by
	\begin{align}
		q_i(s,\bt)=(-1)^{s}\frac{\langle s+k+1, -s|e^{H(\bt)}\beta_i g|k\rangle}{\langle s+k, -s|e^{H(\bt)}g|k\rangle}, \
		r_i(s,\bt)=(-1)^{s+1}\frac{\langle s+k, -s-1|e^{H(\bt)}\beta_i^{*}g|k\rangle}{\langle s+k+1, -s-1|e^{H(\bt)}g|k\rangle},\label{defqiri}
	\end{align}
	then the following relations hold
	\begin{align}
		&\frac{\langle s+k+2, -s|e^{H(\bt)}\beta_2\beta_1 g|k\rangle}{\langle s+k, -s|e^{H(\bt)}g|k\rangle}=q_1\cdot \Delta(q_2)-\Delta(q_1)\cdot q_2
		,\label{vac21} \\
		&\frac{\langle s+k,-s|e^{H(\bt)}\beta_2^{*}\beta_1 g|k\rangle }{\langle s+k,-s|e^{H(\bt)}g|k\rangle}=-\Delta^{-1}(q_1r_2),\label{vac2*1}\\
		&\frac{\langle s+k,-s|e^{H(\bt)}\beta_2\beta_1^{*} g|k\rangle }{\langle s+k,-s|e^{H(\bt)}g|k\rangle}=\Delta^{-1}(r_1q_2),\label{vac21*}\\
		&\frac{\langle s+k-2, -s|e^{H(\bt)}\beta_2^{*}\beta_1^{*}g|k\rangle}{\langle s+k, -s|e^{H(\bt)}g|k\rangle}=-\Lambda^{-2}\Big(r_1\cdot\Delta(r_2)-\Delta(r_1)\cdot r_2 \Big).\label{vac22}
	\end{align}
\end{lemma}
\begin{proof}
	We only prove \eqref{vac21} and \eqref{vac2*1}, since \eqref{vac21*} and \eqref{vac22} can be obtained in a similar way. Firstly, we denote
	\begin{align}\label{ABstexpression}
		A(s,\bt)=(-1)^{s}\frac{\langle s+k+2,-s|e^{H(\bt)}\beta_2\beta_1g|k\rangle}{\langle s+k+1,-s|e^{H(\bt)}\beta_1g|k\rangle},\ 
		B(s,\bt)=(-1)^{s+1}\frac{\langle s+k+1,-s-1|e^{H(\bt)}\beta_2^{*}\beta_1g|k\rangle }{\langle s+k+2,-s-1|e^{H(\bt)}\beta_1g|k\rangle}.
	\end{align}
 As we can know, there exist unique $a^{(i)}(z)$ and $ b^{(i)}(z) \in\mathbb{C}((z^{-1}))$ such that 
	 $\beta_2$ and $\beta_2^{*}$ can be written into
	\begin{align*}
		&\beta_2={\rm Res}_za^{(1)}(z)\psi^{(1)}(z)+{\rm Res}_za^{(2)}(z)\psi^{(2)}(z),\\
		&\beta^{*}_2={\rm Res}_zb^{(1)}(z)\psi^{(1)*}(z)+{\rm Res}_zb^{(2)}(z)\psi^{(2)*}(z),
	\end{align*} 
where $\psi^{(i)}(z)$ and $\psi^{(i)*}(z)$ are defined by \eqref{defpsi(1)and*}. Therefore $A(s,\bt)$ and $B(s,\bt)$ can be rewritten into 
\begin{align*}
	&A(s,\bt)=\oint_{C_R}\frac{{\rm d}z}{2\pi i}z^{k+1}a^{(1)}(z)\psi_1^{[1]}(s,\bt,z)+\oint_{C_r}\frac{{\rm d}z}{2\pi i}z^{-1}a^{(2)}(z^{-1})\psi_2^{[1]}(s,\bt,z),\\
	&B(s,\bt)=\oint_{C_R}\frac{{\rm d}z}{2\pi i}z^{-k-1}b^{(1)}(z)\psi_1^{*[1]}(s,\bt,z)-\oint_{C_r}\frac{{\rm d}z}{2\pi i}z^{-3}b^{(2)}(z^{-1})\psi_2^{*[1]}(s,\bt,z),
\end{align*}
where 
	\begin{equation*}
	\begin{aligned}
		&\psi_1^{[1]}(s,\bt,z)=(-1)^{s}z^{-k-1}\frac{\langle s+k+2,-s|e^{H(\bt)}\psi^{(1)}(z)\beta_1 g|k\rangle }{\langle s+k+1,-s|e^{H(\bt)}\beta_1 g|k\rangle },\\
		&\psi_2^{[1]}(s,\bt,z)=(-1)^{s}z^{-1}\frac{\langle s+k+2,-s|e^{H(\bt)}\psi^{(2)}(z^{-1})\beta_1 g|k\rangle
		}{\langle s+k+1,-s|e^{H(\bt)}\beta_1 g|k\rangle }
		\\
		&\psi_1^{*[1]}(s,\bt,z)=(-1)^{s+1}z^{k+1}\frac{\langle s+k+1, -s-1|e^{H(\bt)}\psi^{(1)*}(z)\beta_1 g|k\rangle}{\langle s+k+2, -s-1 |e^{H(\bt)}\beta_1 g|k\rangle} \\
		&\psi_2^{*[1]}(s,\bt,z)=(-1)^{s}z\frac{\langle s+k+1, -s-1|e^{H(\bt)}\psi^{(2)*}(z^{-1})\beta_1 g|k\rangle}{\langle s+k+2, -s-1 |e^{H(\bt)}\beta_1 g|k\rangle}.
	\end{aligned}
\end{equation*}
And $q_2$, $r_2$ can also be rewritten into
\begin{equation}\label{q2r2resab}
	 \begin{aligned}
		&q_2(s,\bt)=\oint_{C_R}\frac{{\rm d}z}{2\pi i}z^{k}a^{(1)}(z)\psi_1(s,\bt,z)+\oint_{C_r}\frac{{\rm d}z}{2\pi i}z^{-1}a^{(2)}(z^{-1})\psi_2(s,\bt,z),\\
		&r_2(s,\bt)=\oint_{C_R}\frac{{\rm d}z}{2\pi i}z^{-k}b^{(1)}(z)\psi_1^{*}(s,\bt,z)-\oint_{C_r}\frac{{\rm d}z}{2\pi i}z^{-3}b^{(2)}(z^{-1})\psi_2^{*}(s,\bt,z),
	\end{aligned}
\end{equation}
where $\psi_i$ and $\psi_i^{*}$ are defined by \eqref{defwav}. Then according to Proposition \ref{traphipsi}, we can find
\begin{align*}
	&\psi_1^{[1]}=z^{-1}\cdot T_{+}(q_1)(\psi_1),\qquad \quad \psi_2^{[1]}=T_{+}(q_1)(\psi_2),\\ &\psi_1^{[1]*}=z\cdot \iota_{\Lambda^{+}}(T_{+}(q_1)^{-1})^{*}(\psi_1^{*}),\ \ \psi_2^{[1]*}= \iota_{\Lambda^{-}}(T_{+}(q_1)^{-1})^{*}(\psi_2^{*}),
\end{align*}
So further by \eqref{q2r2resab}, one can obtain 
\begin{equation*}\label{ATqBTr}
	A(s,\bt)=T_{+}(q_1)(q_2),\ B(s,\bt)=(T_{+}(q_1)^{-1})^{*}(r_2),
\end{equation*}
which finish the corresponding proof.
\end{proof}
\begin{remark}
	$A(s,\bt)$ and $B(s,\bt)$ in \eqref{ABstexpression} are respectively Toda eigenfunction and adjoint eigenfunction corresponding to $\tau^{[1]}(s,\bt)=(-1)^{s}q_1(s,\bt)\tau(s,\bt)$.
\end{remark}
After above preparation, now we can give determinant formulas for VEV \eqref{Todavev}, which is just the theorem below.
\begin{theorem}\label{theorem:vev}
	Given $\beta_i\in\oplus_{l\in\mathbb{Z}}\psi_l$, $\beta^{*}_j\in\oplus_{l\in\mathbb{Z}}\psi_l^{*}$, if  define $q_i$ and $r_j$ in \eqref{defqiri},
	then 
	\begin{itemize}
		\item for $n\geq m$,
	\end{itemize}
	\begin{align}\label{vacimp2}
		\frac{\langle s+k+n-m,-s|e^{H(\bt)}\beta_m^{*}\cdots\beta_1^{*}\beta_n\cdots\beta_1g|k\rangle}{\langle s+k,-s|e^{H(\bt)}g|k\rangle }=(-1)^{nm+(m+n)s}
		IW_{m,n}(r_m,\dots,r_1;q_1,\dots,q_n),
	\end{align}
	\begin{itemize}
		\item for $n<m$,
	\end{itemize}
	\begin{align}\label{vacimp2*}
		\frac{\langle s+k+n-m,-s|e^{H(\bt)}\beta_m^{*}\cdots\beta_1^{*}\beta_n\cdots\beta_1g|k\rangle}{\langle s+k,-s|e^{H(\bt)}g|k\rangle}=(-1)^{(n+m)s}\Lambda^{-1}\left(IW_{n,m}^{*}(q_n,\dots,q_1;r_1,\dots,r_m) \right).
	\end{align}

\end{theorem}
\begin{proof}
If denote $A_{i,j}(s,\bt)=\langle s+k+j-i,-s|e^{H(\bt)}\beta_i^{*}\cdots\beta_1^{*}\beta_j\cdots\beta_1g|k\rangle$ $(0\leq i\leq m;0\leq j\leq n)$, then $A_{m,n}$ is just the VEV \eqref{Todavev}. Note that we can write $A_{m,n}$ into the form below
	\begin{align*}
A_{m,n}=\frac{A_{m,n}}{A_{m-1,n}}\cdot\frac{A_{m-1,n}}{A_{m-2,n}}\cdots\frac{A_{1,n}}{A_{0,n}}\cdot\frac{A_{0,n}}{A_{0,n-1}}\cdots\frac{A_{0,1}}{A_{0,0}}\cdot A_{0,0}.
	\end{align*}
Since $\beta_i^{*}\cdots\beta_1^{*}\beta_j\cdots\beta_1g|k\rangle\in\mathcal{F}$
satisfies fermionic KP hierarchy \eqref{ferKP}, there exists one fixed integer $l_{ij}\in\mathbb{Z}$ and $g_{ij}\in GL_\infty$ such that $\beta_i^{*}\cdots\beta_1^{*}\beta_j\cdots\beta_1g|k\rangle
=g_{ij}|l_{ij}\rangle$ (please refer to \cite{Liu2024,Miwa2000}). Thus now we can obtain by Lemma \ref{covacqr} that
\begin{align*}
		&\frac{A_{0,j}(s,\bt) }{A_{0,j-1}(s,\bt) }=(-1)^{s}\cdot q_j^{[j-1]}(s,\bt),\ (1\leq j\leq n),  \\
		&\frac{A_{i,n}(s,\bt)}{A_{i-1,n}(s,\bt)}=(-1)^{s}\cdot\Lambda^{-1}\left( r_i^{[n+i-1]}(s,\bt)\right),\ (1\leq i\leq m),
	\end{align*}
which imply that 
\begin{equation*}
A_{m,n}=(-1)^{(n+m)s}\Lambda^{-1}(r_m^{[n+m-1]})\cdots\Lambda^{-1}(r_1^{[n]})q_n^{[n-1]}\cdots q_2^{[1]}\cdot q_1\cdot A_{0,0}.
\end{equation*}
	Note that by \eqref{deftraqrtau},  we can regard $(-1)^{\frac{s(s-1)}{2}}A_{m,n}$ as ($n+m$)--step Darboux transformed Toda tau function $\tau^{[n+m]}$ derived from $\tau^{[0]}=(-1)^{\frac{s(s-1)}{2}}A_{0,0}$. Then by using \eqref{taunmIwmn} and \eqref{taunmIW*nm}, we can obtain \eqref{vacimp2} and \eqref{vacimp2*} respectively. 
\end{proof}
\begin{remark}
If we disrupt the order of $\beta_i$ and $\beta_j^*$ in 
$\langle s+k+n-m,-s|e^{H(\bt)}\beta_m^{*}\cdots\beta_1^{*}
\beta_n\cdots\beta_1g|k\rangle$, then the corresponding results remain the same except the sign. If we swap the positions of two adjacent items in $\beta_m^{*}\cdots\beta_1^{*}
\beta_n\cdots\beta_1$, a minus sign will appear, which can help to determine the sign. The reason why we can do this is the commutativity of $T_\pm$. As for $\beta_i\beta_j^*=-\beta_j^*\beta_i+c_{ij}$, note that the constant $c_{ij}$ is absorbed in $\Delta^{-1}(q_ir_j)$. In practice, we always represent $\beta_i$ and $\beta_j^*$ using $\psi^{(a)}(\lambda_i)$ and $\psi^{*(b)}(\mu_j)$ respectively. Notice that $\psi^{(a)}(\lambda_i)$ anti--commute with $\psi^{*(b)}(\mu_j)$ when $\lambda_i\neq\mu_j$. Therefore we believe the results in Theorem \ref{theorem:vev} can be used widely. 
\end{remark}
Based on this remark, we can also get determinant formulas for $\langle s+k+n-m,-s|e^{H(\bt)}\beta_{n}\cdots\beta_1\beta_{m}^{*}\cdots\beta_{1}^{*}g|k\rangle$, which are given as follows.
	\begin{itemize}
		\item for $n\geq m$,
	\end{itemize}
	\begin{align*}
		\frac{\langle s+k+n-m,-s|e^{H(\bt)}\beta_{n}\cdots\beta_1\beta_{m}^{*}\cdots\beta_{1}^{*}g|k\rangle}{\langle s+k,-s|e^{H(\bt)}g|k\rangle}=(-1)^{(m+n)s}
		IW_{m,n}(r_m,\dots,r_1;q_1,\dots,q_n),
	\end{align*}
	\begin{itemize}
		\item for $n<m$,
	\end{itemize}
	\begin{align*}
		\frac{\langle s+k+n-m,-s|e^{H(\bt)}\beta_{n}\cdots\beta_1\beta_{m}^{*}\cdots\beta_{1}^{*}g|k\rangle}{\langle s+k,-s|e^{H(\bt)}g|k\rangle}=(-1)^{nm+(n+m)s}\Lambda^{-1}\left(IW_{n,m}^{*}(q_n,\dots,q_1;r_1,\dots,r_m) \right).
	\end{align*}
Particularly, by setting $m=0$ in \eqref{vacimp2} and $n=0$ in \eqref{vacimp2*}, we get the following corollary.
\begin{corollary}Given $\beta_i\in\oplus_{l\in\mathbb{Z}}\psi_l$, $\beta^{*}_j\in\oplus_{l\in\mathbb{Z}}\psi_l^{*}$, $q_i$ and $r_j$ defined in \eqref{defqiri},
	\begin{align*}
		&\frac{\langle s+k+n,-s|e^{H(\bt)}\beta_n\cdots\beta_1g|k\rangle}{\langle s+k,-s|e^{H(\bt)}g|k\rangle}=(-1)^{sn}det\left(\Delta^{i-1}(q_{j}(s))\right)_{1\leq i,j\leq n},\\
		&\frac{\langle s+k-m,-s|e^{H(\bt)}\beta^{*}_{m}\cdots\beta^{*}_{1}g |k\rangle}{\langle s+k,-s|e^{H(\bt)}g|k\rangle}=(-1)^{sm}det\left((\Delta^{*})^{i-1}(r_{j}(s-1))\right)_{1\leq i,j\leq m}.
	\end{align*}
\end{corollary}

Finally, we give some results for VEV \eqref{kpvev} related with KP hierarchy. Using a similar method, we have the following lemma.
\begin{lemma}
Given $g\in GL_\infty$, $\alpha_i\in \oplus_{j\in\mathbb{Z}}\mathbb{C}\psi_j$, $\alpha_i^*\in \oplus_{j\in\mathbb{Z}}\mathbb{C}\psi_j^*$ $(i=0,1,2)$, we have the following relations hold
\begin{align*}
&\frac{\langle k| e^{\mathcal{H}(\bm{x})}\alpha_0\alpha_0^*g|k\rangle}
{\langle k| e^{\mathcal{H}(\bm{x})} g|k\rangle}=\int \textbf{q}_0(\bm{x})\textbf{r}_0(\bm{x})dx\triangleq\Omega(\mathbf{q}_0,\mathbf{r}_0),\\
&\frac{\langle k+2| e^{\mathcal{H}(\bm{x})}\alpha_2\alpha_1 g|k\rangle}
{\langle k| e^{\mathcal{H}(\bm{x})} g|k\rangle}=\textbf{q}_1\textbf{q}_{2,x_1}-\textbf{q}_2\textbf{q}_{1,x_1},\\
&\frac{\langle k-2| e^{\mathcal{H}(\bm{x})}\alpha_2^*\alpha_1^*g|k\rangle}
{\langle k| e^{\mathcal{H}(\bm{x})} g|k\rangle}=-\textbf{r}_1\textbf{r}_{2,x_1}+\textbf{r}_2\textbf{r}_{1,x_1},
\end{align*}
where $k\in\mathbb{Z}$, $\bm{x}=(x_1,x_2,\cdots,)$, $\mathcal{H}(\bm{x})=\sum_{l=1}^{\infty}
\sum_{j\in\mathbb{Z}}x_l :\psi_j\psi^{*}_{j+l}:$ and
\begin{align*}
\textbf{q}_i(\bm{x})=\frac{\langle k+1|e^{\mathcal{H}(\bm{x})}\alpha_ig|k\rangle}{\langle k|e^{\mathcal{H}(\bm{x})}g|k\rangle},\quad \textbf{r}_i(\bm{x})=\frac{\langle k-1|e^{\mathcal{H}(\bm{x})}\alpha^{*}_ig|k\rangle}{\langle k|e^{\mathcal{H}(\bm{x})}g|k\rangle}.
\end{align*}
\end{lemma} 
Then by this lemma, we can also view below VEV 
$$\langle n-m|e^{\mathcal{H}(\bm{x})}\alpha_m^*\cdots\alpha_1^*\alpha_n\cdots\alpha_1g|0\rangle$$
as the transformed KP tau functions under multi--step KP Darboux transformations \cite{jshe,Oevel1993}. Therefore, we have the following results.
\begin{proposition}\label{Pro_VEV_KP}
For $g\in GL_\infty$, $\alpha_i\in \oplus_{j\in\mathbb{Z}}\mathbb{C}\psi_j$, $\alpha_i^*\in \oplus_{j\in\mathbb{Z}}\mathbb{C}\psi_j^*$, we have the following VEV determinant formulas.
\begin{itemize}
  \item $n\geq m$,
\end{itemize}
\begin{align*}
&\frac{\langle n-m|e^{\mathcal{H}(\bm{x})}\alpha_m^*\cdots\alpha_1^*\alpha_n\cdots\alpha_1g|0\rangle}{\langle 0| e^{\mathcal{H}(\bm{x})} g|0\rangle}
=(-1)^{mn}\left|
    \begin{array}{ccc}
      \Omega(\textbf{q}_1,\textbf{r}_m) & \cdots & \Omega(\textbf{q}_n,\textbf{r}_m) \\
      \vdots & \vdots & \vdots \\
      \Omega(\textbf{q}_1,\textbf{r}_1) & \cdots & \Omega(\textbf{q}_n,\textbf{r}_1) \\
      \textbf{q}_1 & \cdots & \textbf{q}_n \\
      \vdots & \vdots & \vdots \\
     \textbf{q}_1^{(n-m-1)} & \cdots & \textbf{q}_n^{(n-m-1)} \\
    \end{array}
  \right|,
\end{align*}
where $f(\bm{x})^{(p)}=\partial_{x_1}^p(f(\bm{x}))$.
\begin{itemize}
  \item $n<m$,
  \begin{align*}
  	\frac{\langle n-m|e^{\mathcal{H}(\bm{x})}\alpha_m^*\cdots\alpha_1^*\alpha_n\cdots\alpha_1g|0\rangle}{\langle 0| e^{\mathcal{H}(\bm{x})} g|0\rangle}=(-1)^{mn+\frac{m(m-1)}{2}+\frac{n(n-1)}{2}}\left|
    \begin{array}{ccc}
      \Omega(\textbf{q}_n,\textbf{r}_1) & \cdots & \Omega(\textbf{q}_n,\textbf{r}_m) \\
      \vdots & \vdots & \vdots \\
      \Omega(\textbf{q}_1,\textbf{r}_1) & \cdots & \Omega(\textbf{q}_1,\textbf{r}_m) \\
      \textbf{r}_1 & \cdots & \textbf{r}_m \\
      \vdots & \vdots & \vdots \\
     \textbf{r}_1^{(m-n-1)} & \cdots & \textbf{r}_m^{(m-n-1)} \\
    \end{array}
  \right|.
  \end{align*}
\end{itemize}
\end{proposition}
\section{Conclusion and discussion}\label{sec:5}
In this paper, we firstly obtain the elementary Darboux transformation operators for the whole 2--Toda hierarchy, that is,
$$T_{+}(q)=\Lambda(q)\cdot\Delta\cdot q^{-1},\quad T_{-}(r)=\Lambda^{-1}(r)^{-1}\cdot\Delta\cdot r.$$
As far as we can know, although Darboux transformations for Toda lattice equations have been extensively investigated \cite{Nimmo1997}, there is little literature that discusses the Darboux transformations for the whole 2--Toda hierarchy. Then the VEV 
$$(-1)^{\frac{s(s-1)}{2}}
\langle s+k+n-m,-s|e^{H(\bt)}\beta_{n}\cdots\beta_1\beta_{m}^{*}\cdots\beta_{1}^{*}g|k\rangle$$
can be viewed as the transformed Toda tau function under multi--Darboux transformation $T^{[n,m]}$ from the initial tau function $\tau^{[0]}(\bt)=(-1)^{\frac{s(s-1)}{2}}
\langle s+k,-s|e^{H(\bt)}g|k\rangle$. Next Based upon this, determinant formulas for the above VEV are obtained in Theorem \ref{theorem:vev}.  Finally we discuss the VEV related to KP hierarchy, that is, 
$$\langle n-m|e^{\mathcal{H}(\bm{x})}\beta_m^*\cdots\beta_1^*\beta_n\cdots\beta_1g|0\rangle,$$
and the corresponding resuts are given in Proposition \ref{Pro_VEV_KP}. By now, we have given the determinant formulas for VEVs related with KP and Toda hierarchies by using Darboux transformations, which are just the methods of integrable systems. We believe above results can be generalized to the general case, that is, VEV
\eqref{p-KPvev}. And we will discuss this question in future.

\section*{Appendix}
The aim of this appendix is to show the determinant representation of
 multi--step Toda Darboux transformation defined by
 \begin{equation*}
 	T^{[n,m]}=T_{-}(r_m^{[n+m-1]})\cdots T_{-}(r_1^{[n]})\cdots T_{+}(q_n^{[n-1]})\cdots T_{+}(q_1),\ (n,m\in\mathbb{N}),
 \end{equation*}  
 where $T_{+}(q)=\Lambda(q)\cdot \Delta\cdot q^{-1}$ and $T_{-}(r)=\Lambda^{-1}(r)^{-1}\cdot\Delta^{-1}\cdot r$ are Toda Darboux transformations, and $q_i$ and $r_j$ are Toda eigenfunctions and adjoint eigenfunctions, respectively.
  
 If introduce general discrete Wronskian determinants as follows,
 \begin{align*}
 	&IW_{m,n}(r_m,\dots,r_1;q_1,\dots,q_{n})=
 	\begin{vmatrix}
 		\Delta^{-1}(q_1r_m)& \Delta^{-1}(q_2r_m)&\cdots&\Delta^{-1}(q_nr_m)\\
 		\vdots &  \vdots&\ddots & \vdots\\
 		\Delta^{-1}(q_1r_1)&\Delta^{-1}(q_2r_1) &\cdots&\Delta^{-1}(q_nr_1)\\
 		q_1&q_2&\cdots&q_n\\
 		\Delta(q_1)  &\Delta(q_2)&\cdots&\Delta(q_n)\\
 		\vdots & \vdots& \ddots & \vdots\\
 		\Delta^{n-m-1}(q_1)&\Delta^{n-m-1}(q_2) &\cdots& \Delta^{n-m-1}(q_n)
 	\end{vmatrix},  \ (n\geq m),   \\
 	&IW^{*}_{n,m}(q_n,\dots,q_1;r_1,\dots,r_m)=
 	\begin{vmatrix}
 		(\Delta^*)^{-1}(r_1q_n)&(\Delta^*)^{-1}(r_2q_n)&\cdots&(\Delta^*)^{-1}(r_mq_n)\\
 		\vdots &  \vdots&\ddots & \vdots\\
 		(\Delta^*)^{-1}(r_1q_1)&(\Delta^*)^{-1}(r_2q_1)&\cdots&(\Delta^*)^{-1}(r_mq_1)\\
 		r_1&r_2&\cdots&r_m\\
 		\Delta^*(r_1)& \Delta^*(r_2)&\cdots& \Delta^*(r_m)\\
 		\vdots &  \vdots&\ddots & \vdots\\
 		(\Delta^*)^{m-n-1}(r_1)& 	(\Delta^*)^{m-n-1}(r_2)&\cdots& 	(\Delta^*)^{m-n-1}(r_m)
 	\end{vmatrix},\ (n<m),
 \end{align*}
 then the determinant representations of $T^{[n,m]}$ and $(T^{[n,m]})^{-1}$ are given in the way below (see \cite{Liu2010,Song2023} for more details).
 \begin{itemize}
 	\item for $n\geq m$,
 	\begin{align}
 		&T^{[n,m]}=\frac{1}{IW_{m,n}(r_m,\dots,r_1;q_1,\dots,q_{n})}\begin{vmatrix}
 			\Delta^{-1}(q_1r_m)& \cdots&\Delta^{-1}(q_nr_m)&\Delta^{-1}r_m\\
 			\vdots &  \ddots & \vdots&\vdots\\
 			\Delta^{-1}(q_1r_1) &\cdots&\Delta^{-1}(q_nr_1)&\Delta^{-1}r_1 \\
 			q_1&\cdots&q_n&1\\
 			\Delta(q_1)  &\cdots&\Delta(q_n)&\Delta\\
 			\vdots & \ddots & \vdots&\vdots\\
 			\Delta^{n-m}(q_1) &\cdots& \Delta^{n-m}(q_n)&\Delta^{n-m}
 		\end{vmatrix}, \label{detT1} \\
 		&(T^{[n,m]})^{-1}=\begin{vmatrix}
 			q_1\Delta^{-1}&\Lambda\Delta^{-1}(q_1r_m)&\cdots&\Lambda\Delta^{-1}(q_1r_1)&\Lambda(q_1)&\cdots&\Lambda\Delta^{n-m-2}(q_1)\\
 			q_2\Delta^{-1}&\Lambda\Delta^{-1}(q_2r_m)&\cdots&\Lambda\Delta^{-1}(q_2r_1)&\Lambda(q_2)&\cdots&\Lambda\Delta^{n-m-2}(q_2)\\
 			\vdots&\vdots&\ddots&\vdots&\vdots&\ddots&\vdots\\
 			q_n\Delta^{-1}&\Lambda\Delta^{-1}(q_nr_m)&\cdots&\Lambda\Delta^{-1}(q_nr_1)&\Lambda(q_n)&\cdots&\Lambda\Delta^{n-m-2}(q_n)
 		\end{vmatrix}\times\nonumber\\
 		&\qquad\qquad \times\frac{(-1)^{n-1}}{\Lambda\left(IW_{m,n}(r_m,\dots,r_1;q_1,\dots,q_{n})\right)},\label{detT2}
 	\end{align}
 \end{itemize} 
 \begin{itemize}
 	\item for $n=m$,
 	\begin{align}
 		&T^{[n,n]}=\frac{1}{IW_{n,n}(r_n,\dots,r_1;q_1,\dots,q_n)}\begin{vmatrix}
 			\Delta^{-1}(q_1r_n)&\cdots&	\Delta^{-1}(q_nr_n)&\Delta^{-1}r_n\\
 			\vdots&\ddots&\vdots&\vdots\\
 			\Delta^{-1}(q_1r_1)&\cdots&	\Delta^{-1}(q_nr_1)&\Delta^{-1}r_1\\
 			q_1&\cdots&q_n&1
 		\end{vmatrix},\label{nemT}\\
 		&(T^{[n,n]})^{-1}=\begin{vmatrix}
 			-1&r_n&\cdots&r_1\\
 			q_1\Delta^{-1}&\Lambda\Delta^{-1}(q_1r_n)&\cdots&\Lambda\Delta^{-1}(q_1r_1)\\
 			\vdots&\vdots&\ddots&\vdots\\
 			q_n\Delta^{-1}&\Lambda\Delta^{-1}(q_nr_n)&\cdots&\Lambda\Delta^{-1}(q_nr_1)
 		\end{vmatrix}
 		\frac{-1}{\Lambda\left(IW_{n,n}(r_n,\dots,r_1;q_1,\dots,q_n) \right)},\label{nemTi}
 	\end{align}
 \end{itemize}
 \begin{itemize}
 	\item for $n<m$,
 	\begin{align}
 		&T^{[n,m]}=\frac{1}{\Lambda^{-1}\left( IW^{*}_{n,m}(q_n,\dots,q_1;r_1,\dots,r_m) \right)}\times\nonumber\\
 		&\qquad\ \
 		\times\begin{vmatrix}
 			\Lambda^{-1}(\Delta^{*})^{-1}(r_1q_n)&\cdots&\Lambda^{-1}(\Delta^{*})^{-1}(r_1q_1)&\Lambda^{-1}(r_1)&\cdots&\Lambda^{-1}(\Delta^{*})^{m-n-2}(r_1)&\Delta^{-1}r_1\\
 			\vdots&\ddots &\vdots &\vdots &\ddots &\vdots &\vdots\\
 			\Lambda^{-1}(\Delta^{*})^{-1}(r_mq_n)&\cdots&\Lambda^{-1}(\Delta^{*})^{-1}(r_mq_1)&\Lambda^{-1}(r_m)&\cdots&\Lambda^{-1}(\Delta^{*})^{m-n-2}(r_m)&\Delta^{-1}r_m
 		\end{vmatrix},\label{nlmT}\\
 		&(T^{[n,m]})^{-1}=\begin{vmatrix}
 			q_n\Delta^{-1}&(\Delta^{*})^{-1}(r_1q_n)&\cdots&(\Delta^{*})^{-1}(r_mq_n)\\
 			\vdots&\vdots&\ddots&\vdots\\
 			q_1\Delta^{-1}&(\Delta^{*})^{-1}(r_1q_1)&\cdots&(\Delta^{*})^{-1}(r_mq_1)\\
 			1&r_1&\cdots&r_m\\
 			\Delta&\Delta^{*}(r_1)&\cdots&\Delta^{*}(r_m)\\
 			\vdots&\vdots&\ddots&\vdots\\
 			\Delta^{m-1}&(\Delta^{*})^{m-n}(r_1)&\cdots&(\Delta^{*})^{m-n}(r_m)
 		\end{vmatrix}\frac{(-1)^{m}}{IW^{*}_{n,m}(q_n,\dots,q_1;r_1,\dots,r_m)},\label{nlmTi}
 	\end{align}
 \end{itemize}
 where the determinants of \eqref{detT1}, \eqref{nemT}, and \eqref{nlmT} are expanded using their last columns, with the elements appearing on the right-hand side, whereas those of \eqref{detT2}, \eqref{nemTi}, and \eqref{nlmTi} are expanded by their first columns, with the elements appearing on the left-hand side.\\
 
\noindent{\bf Acknowledgements}: { This work is supported by the National Natural Science Foundations of China (Nos. 12171472 and 12261072) and ``Qinglan Project" of Jiangsu Universities. }\\

\noindent{\bf Conflict of Interest}:
The author have no conflicts to disclose.\\

\noindent{\bf Data availability}: Date sharing is not applicable to this article as no new data were created or analyzed in this study.\\

\end{document}